\documentclass[11pt]{article}
\usepackage{fullpage}
\usepackage{amsmath,amsfonts,amsthm,amssymb}
\usepackage{url}
\usepackage{algorithm}
\usepackage{algorithmic}
\usepackage{bbm}
\usepackage{float}
\usepackage{framed}
\usepackage{enumerate}

\usepackage{color}
\usepackage[usenames,dvipsnames,svgnames,table]{xcolor}
\usepackage[colorlinks=true,
            linkcolor=red,
            urlcolor=black,
            citecolor=gray]{hyperref}

\DeclareMathOperator*{\E}{\mathbb{E}}
\let\Pr\relax
\DeclareMathOperator*{\Pr}{\mathbb{P}}

\DeclareMathOperator*{\rank}{rank}
\DeclareMathOperator*{\tr}{tr}
\let\ker\relax
\DeclareMathOperator*{\ker}{ker}

\newcommand{\wh}{\widehat}

\newcommand{\eqdef}{\mathbin{\stackrel{\rm def}{=}}}
\newcommand{\norm}[1]{\|#1\|}
\newcommand{\bs}[1]{\boldsymbol{#1}}
\newcommand{\bv}[1]{\mathbf{#1}}

\newcommand{\refine}{\texttt{RefineSparsifier}}
\newcommand{\maintain}{\texttt{MaintainSketches}}
\newcommand{\maintainma}{\texttt{MaintainMatrixSketches}}
\newcommand{\sample}{\texttt{RowSampleMatrix}}

\newcommand{\plog}{\mathop\mathrm{polylog}}
\newcommand{\poly}{\mathop\mathrm{poly}}
\newcommand{\R}{\mathbb{R}}
\newcommand{\mvar}[1]{\bv{#1}}
\newcommand{\ma}{\mvar{A}}
\newcommand{\mb}{\mvar{B}}

\newcommand{\mf}{\mvar{F}}
\newcommand{\mk}{\mvar{K}}

\newcommand{\ms}{\mvar{S}}
\newcommand{\mw}{\mvar{W}}
\newcommand{\mpi}{\mvar{\Pi}}

\newcommand{\dist}{\mathcal{D}}

\makeatletter
\newtheorem*{rep@theorem}{\rep@title}
\newcommand{\newreptheorem}[2]{%
\newenvironment{rep#1}[1]{%
 \def\rep@title{#2 \ref{##1}}%
 \begin{rep@theorem}}%
 {\end{rep@theorem}}}
\makeatother

\newtheorem{theorem}{Theorem}
\newtheorem*{theorem*}{Theorem}
\newreptheorem{theorem}{Theorem}
\newtheorem{lemma}{Lemma}
\newtheorem*{lemma*}{Lemma}
\newreptheorem{lemma}{Lemma}

  \usepackage{nth}
  \usepackage{intcalc}

\title{Single Pass Spectral Sparsification in Dynamic Streams}

\author{\and\and
Michael Kapralov\\MIT\\ {kapralov@mit.edu}
\and
Yin Tat Lee\\MIT\\ {yintat@mit.edu}
\and
Cameron Musco\\MIT\\ {cnmusco@mit.edu}
\and\and
Christopher Musco\\MIT\\ {cpmusco@mit.edu}
\and
Aaron Sidford\\MIT\\ {sidford@mit.edu}
}
\date{}

\begin{document}
\maketitle
\begin{abstract}
We present the first single pass algorithm for computing spectral sparsifiers of graphs in the dynamic semi-streaming model. Given a single pass over a stream containing insertions and deletions of edges to a graph $G$, our algorithm maintains a randomized linear sketch of the incidence matrix of $G$ into dimension $O(\frac{1}{\epsilon^2}n\plog (n))$.
Using this sketch, at any point, the algorithm can output a $(1 \pm \epsilon)$ spectral sparsifier for $G$ with high probability. 

While $O(\frac{1}{\epsilon^2} n \plog(n))$ space algorithms are known for computing \emph{cut sparsifiers} in dynamic streams [AGM12b, GKP12] and spectral sparsifiers in \emph{insertion-only} streams [KL11], prior to our work, the best known single pass algorithm for maintaining spectral sparsifiers in dynamic streams required sketches  of dimension $\Omega(\frac{1}{\epsilon^2}n^{5/3})$ [AGM13].

To achieve our result, we show that, using a coarse sparsifier of  $G$ and a linear sketch of $G$'s incidence matrix, it is possible to sample edges by effective resistance,  obtaining a spectral sparsifier of arbitrary precision. Sampling from the sketch requires a novel application of $\ell_2/\ell_2$ sparse recovery, a natural extension of the $\ell_0$ methods used for cut sparsifiers in [AGM12b]. Recent work of [MP12] on row sampling for matrix approximation gives a recursive approach for obtaining the required coarse sparsifiers. 

Under certain restrictions, our approach also extends to the problem of maintaining a spectral approximation for a general matrix $A^\top A$ given a stream of updates to rows in $A$.

\end{abstract}
\thispagestyle{empty}

\clearpage
\setcounter{page}{1}

\section{Introduction}
\subsection{The Dynamic Semi-Streaming Model}
When processing massive graph datasets arising from social networks, web topologies, or interaction graphs, computation may be as limited by space as it is by runtime. To cope with this issue, one might hope to apply techniques
from the streaming model of computation, which restricts algorithms to few passes over the input and space polylogarithmic in the input size. Streaming algorithms
have been studied extensively in various application domains -- see \cite{muthukrishnan2005data} for an overview. However, the model has proven too restrictive for even the
simplest graph algorithms. For example, testing $s$-$t$ connectivity requires
$\Omega(n)$ space \cite{henz:lb}. 

The less restrictive semi-streaming model, in which the
algorithm is allowed $\tilde O(n)$ space, is more suited for graph
algorithms~\cite{feigenbaum2005graph}, and has received significant attention in recent years. In this model, a processor receives a stream of edges over a fixed set of $n$ nodes.  Ideally, the processor should only have to perform a single pass (or few passes) over the edge stream, and the processing time per edge, as well as the time required to output the final answer, should be small.

In the \emph{dynamic semi-streaming model}, the graph stream may include both edge insertions and deletions \cite{linearMeasurement}. This extension captures the fact that large graphs are unlikely to be static. Dynamic semi-streaming algorithms allow us to quickly process general updates in the form of edge insertions and deletions to maintain a small-space representation of the graph from which we can later compute a result. Sometimes the dynamic model is referred to as the \emph{insertion-deletion model}, in contrast to the more restrictive \emph{insertion-only model}.

Work on semi-streaming algorithms in both the dynamic and insertion-only settings is extensive. Researchers have tackled connectivity, bipartiteness, minimum spanning trees, maximal matchings, and spanners among other problems \cite{feigenbaum2005graph, epstein2011improved, elkin2011streaming, linearMeasurement,gssss}. In \cite{gregorSurvey},  McGregor surveys much of this progress and provides a more complete list of citations.

\subsection{Streaming Sparsification}
There has also been a focus on computing \emph{general purpose} graph compressions in the streaming setting. The goal is to find a subgraph of an input graph $G$ that has significantly fewer edges than $G$, but still maintains important properties of the graph. Hopefully, this \emph{sparsified graph} can be used to approximately answer a variety of questions about $G$ with reduced space and time complexity. Typically, the goal is to find a subgraph with just $O(n \log n)$ edges in comparison to the possible $O(n^2)$ edges in $G$. 

First introduced by Bencz{\'u}r and Karger \cite{benczur1996approximating}, a \emph{cut sparsifier} of a graph $G$ is a weighted subgraph with only $O(\frac{1}{\epsilon^2}n\log n)$ edges that preserves the total edge weight over every cut in $G$ to within a $(1 \pm \epsilon)$ multiplicative factor. Cut sparsifiers can be used to compute approximations for minimum cut, sparsest cut, maximum flow, and a variety of other problems over $G$. In \cite{spielmanTengSpectralSparse}, Spielman and Teng introduce the stronger \emph{spectral sparsifier}, a weighted subgraph whose Laplacian spectrally approximates the Laplacian of $G$. In addition to maintaining the cut approximation of Bencz{\'u}r and Karger, spectral sparsifiers can be used to approximately solve linear systems over the Laplacian of $G$, and to approximate effective resistances, spectral clusterings, random walk properties, and a variety of other computations. 

\if Both cut and spectral sparsifiers are powerful primitives for the semi-streaming model since they are rich compressions of dense graphs that can be stored in $O(n\plog (n))$ space. Streaming algorithms that allow us to extract a sparsifier from a graph immediately yield streaming algorithms for the many problems that sparsifiers allow us to approximate. 
\fi

The problem of computing graph sparsifiers in the semi-streaming model has received a lot of attention. Given just $\tilde{O}(n) = O(n \plog(n))$ space, the hope is to compute a sparsifier using barely more space than required to store the sparsifier, which will typically have $O(n\log n)$ edges.  Ahn and Guha give the first single pass, insertion-only algorithm for cut sparsifiers \cite{ahnStreamingSparsification}. Kelner and Levin give a single pass, insertion-only algorithm for spectral sparsifiers \cite{kelner2011spectral}. Both algorithms store a sparse graph: edges are added as they are streamed in and, when the graph grows too large, it is resparsified. The construction is very clean, but inherently does not extend to the dynamic model since, to handle edge deletions, we need more information than just a sparsifier itself. Edges eliminated to create an intermediate sparsifier may become critically important later if other edges are deleted, so we need to maintain information that allows recovery of such edges.

Ahn, Guha, and McGregor make a very important insight in \cite{linearMeasurement}, demonstrating the power of linear graph sketches in the dynamic model. They present the first dynamic algorithm for cut sparsifiers, which initially required $O(\frac{1}{\epsilon^2} n^{1+\gamma})$ space and $O(1/\gamma)$ passes over the graph stream. However, the result was later improved to a single pass and $O(\frac{1}{\epsilon^2}n\plog (n))$ space \cite{gssss,goel2012single}. Our algorithm extends the sketching and sampling approaches from these papers to the spectral problem.

In \cite{ahn2013spectral}, the authors show that linear graph sketches that capture connectivity information can be used to coarsely approximate spectral properties and they obtain spectral sparsifiers using $O(\frac{1}{\epsilon^2}n^{5/3} \plog(n))$ space in the dynamic setting. However, they also show that their coarse approximations are tight, so a new approach is required to obtain spectral sparsifiers using just $O(\frac{1}{\epsilon^2}n\plog (n))$ space. They conjecture that a dynamic algorithm for doing so exists. The development of such an algorithm is also posed as an open question in \cite{gregorSurvey}. A two-pass algorithm for constructing a spectral sparsifier in the dynamic streaming model using $O \left (\frac1{\epsilon^2}n^{1+o(1)} \right)$ space is presented in~\cite{KW14}. The approach is very different from ours: it leverages a reduction from spanner constructions to spectral sparsification presented in~\cite{KP12}. It is not known if this approach extends to a space efficient single pass algorithm.

\subsection{Our Contribution}

Our main result is an algorithm for maintaining a small graph sketch from which we can recover a spectral sparsifier. For simplicity, we present the algorithm in the case of unweighted graphs. However, in Section \ref{weighted}, we show that it is easily extended to weighted graphs. This model matches what is standard for dynamic cut sparsifiers  \cite{gssss,goel2012single}. 

\begin{theorem}[Main Result]
\label{main_sparsification_theorem} There exists an algorithm that, for any $\epsilon > 0$, processes a  list of edge insertions and deletions for an unweighted graph $G$ in a single pass and maintains a set of linear sketches of this input in $O\left(\frac1{\epsilon^2} n \plog(n)\right)$ space. From these sketches, it is possible to recover, with high probability, a weighted subgraph $H$ with $O(\frac1{\epsilon^2}n\log n)$ edges such that $H$ is a $(1 \pm \epsilon)$ spectral sparsifier of $G$. The algorithm 
recovers $H$ in $O\left(\frac1{\epsilon^2}n^2 \plog (n)\right)$ time.
\end{theorem}

It is well known that independently sampling edges from a graph $G$ according to their \emph{effective resistances} (i.e. leverage scores) gives a $(1\pm \epsilon)$ spectral sparsifier of $G$ with $O(\frac1{\epsilon^2}n \log n)$ edges \cite{graphSparsificationEffectiveResistance}. We can `refine' any coarse sparsifier for $G$ by using it to approximate effective resistances and then resample edges according to these approximate resistances. We show how to perform this refinement in the streaming setting, extending graph sketching techniques initially used for cut sparsifiers (\cite{gssss,goel2012single}) and introducing a new sampling technique based on an $\ell_2$ heavy hitters algorithm. Our refinement procedure is combined with a clever recursive method for obtaining a coarse sparsifier introduced by Miller and Peng in a recent paper on iterative row sampling for matrix approximation \cite{pengV1}.

The fact that our algorithm maintains a linear sketch of the streamed graph allows for the 
simple handling of edge deletions, which are treated as negative edge insertions. Additionally,  due to their linearity, our sketches are composable -- sketches of subgraphs can simply be added to produce a sketch of the full graph. Thus, our techniques are directly applicable in distributed settings where separate processors hold different subgraphs or each processes different edge substreams. 

Our application of linear sketching also gives a nice information theoretic result on graph compression. A spectral sparsifier is a powerful compression for a graph. It maintains, up to an $\epsilon$ factor, all spectral information about the Laplacian using just $O(\frac1{\epsilon^2}n\log n)$ space. At first glance, it may seem that such a compression requires careful analysis of the input graph to determine what information to keep and what to discard. However, the non-adaptive linear sketches used in our algorithm are completely \emph{oblivious}: at each edge insertion or deletion, we do not need to examine the current compression at all to make the appropriate update. As in sparse recovery or dimensionality reduction, we essentially just multiply the vertex edge incidence matrix by a random projection matrix, decreasing its height drastically in the process. Nevertheless, the oblivious compression obtained holds as much information as a spectral sparsifier -- in fact, we show how to extract a spectral sparsifier from it! Furthermore, the compression is only larger than $O(\frac1{\epsilon^2}n\log n)$ by log factors. Our result is the first of this kind in the spectral domain. The only other streaming algorithm for spectral sparsification that uses $O(\frac{1}{\epsilon^2}n\plog (n))$ space is distinctly non-oblivious \cite{kelner2011spectral} and oblivious subspace embeddings for compressing general matrices inherently require $O(n^2 \plog (n))$ space, even when the matrix is sparse (as in the case of an edge vertex incidence matrix) \cite{sarlos2006improved,clarkson2013low, meng2013, osnap}.

Finally, it can be noted that our proofs rely very little on the fact that our data stream represents a graph. We show that, with a few modifications, given a stream of row updates for a general structured matrix $A$, it is possible to maintain a $O(\frac1{\epsilon^2} n \plog(n))$ sized sketch from which a spectral approximation to $A^\top A$ can be recovered. By structured, we mean any matrix whose rows are selected from some fixed dictionary of size $\poly(n)$. Spectral graph sparsification is a special case of this problem: set $A$ to be the vertex edge incidence matrix of our graph. The dictionary is the set of all possible ${n\choose 2}$ edge rows that may appear in $A$ and $A^\top A$ is the graph Laplacian.  



\subsection{Road Map}
\begin{description}
\item[Section \ref{notation}] Lay out notation, build linear algebraic foundations for spectral sparsification, and present  lemmas for graph sampling and sparse recovery required by our algorithm.
\item[Section \ref{algorithm_overview}] Give an overview of our central algorithm, providing intuition and motivation.
\item[Section \ref{recursive_algorithm}] Present an algorithm of Miller and Peng (\cite{pengV1}) for building a chain of coarse sparsifiers and prove our main result, assuming a primitive for sampling edges by effective resistance in the streaming model.
\item[Section \ref{streaming_row_sampling}] Develop this sampling primitive, our main technical contribution.
\item[Section \ref{weighted}] Show how to extend the algorithm to weighted graphs.
\item[Section \ref{structured}] Show how to extend the algorithm to general structured matrices.
\item[Section \ref{pseudorandomness}] Remove our assumption of fully independent hash functions, using a pseudorandom number generator to achieve a final small space algorithm.
\end{description}

\section{Notation and Preliminaries}\label{notation}

\subsection{Graph Notation}
Let $\bv{B}_n \in \mathbb{R}^{{n \choose 2} \times n}$ be the vertex edge incidence matrix of the undirected, unweighted complete graph over $n$ vertices. $\bv{b}_e$, the row corresponding to edge $e=(u,v)$ contains a $1$ in column $u$, a $(-1)$ in column $v$, and $0$'s elsewhere. 

We write the vertex edge incidence matrix of an unweighted, undirected graph $G(V,E)$ as $\bv{B} = \bv{S}\bv{B}_n$ where $\bv{S}$ is an ${n \choose 2} \times {n \choose 2}$ diagonal matrix with ones at positions corresponding to edges contained in $G$ and zeros elsewhere.\footnote{Typically rows of $\bv{B}$ that are all $0$ are removed, but we find this formulation more convenient for our purposes.}
The $n \times n$ Laplacian matrix of $G$ is given by  $\bv{K} =  \bv{B}^\top \bv{B}$. 


\subsection{Spectral Sparsification}
For any matrix $\bv{B} \in \mathbb{R}^{m \times n}$, $\bv{\tilde K}$ is a $(1 \pm \epsilon)$ spectral sparsifier of $\bv{K} = \bv{B}^\top\bv{B}$ if, $\forall \bv{x} \in \mathbb{R}^n$, $(1-\epsilon) \bv{x}^\top \bv{K} \bv{x} \le \bv{x}^\top \bv{\tilde K} \bv{x} \le (1+\epsilon)\bv{x}^\top \bv{K} \bv{x}$. This condition can also be written as $(1-\epsilon)\bv{K} \preceq \bv{\tilde K} \preceq (1+\epsilon)\bv{K}$ where $\bv{C} \preceq \bv{D}$ indicates that $\bv{D} - \bv{C}$ is positive semidefinite. More succinctly, $\bv{\tilde K} \approx_\epsilon \bv{K}$ denotes the same condition. 
We also use the slightly weaker notation $(1-\epsilon)\bv{K} \preceq_{r} \bv{\tilde K} \preceq_r (1+\epsilon)\bv{K}$ to indicate that $(1-\epsilon) \bv{x}^\top \bv{K} \bv{x} \le \bv{x}^\top \bv{\tilde K} \bv{x} \le (1+\epsilon)\bv{x}^\top \bv{K} \bv{x}$ for all $\bv{x}$ in the \emph{row span} of $\bv{K}$. If $\bv{\tilde K}$ has the same row span as $\bv{K}$ this notation is equivalent to the initial notion of spectral sparsification.

While these definitions apply to general matrices, for our purposes, $\bv{B}$ is typically the vertex edge incidence matrix of a graph $G$ and $\bv{K}$ is a graph Laplacian. We do not always require our approximation $\bv{\tilde K}$ to be the graph Laplacian of a weighted subgraph, which is a standard assumption. For this reason, we avoid the standard $\bv{L}_G$ notation for the Laplacian.  For our purposes, $\bv{\tilde K}$ is always be a sparse symmetric diagonally dominant matrix with no more than $O(n\log{n})$ non-zero entries. In fact, it will always be the Laplacian of a sparse subgraph, but possibly with weight added to its diagonal entries. Furthermore, the final approximation returned by our streaming algorithm will be a bonafide spectral graph sparsifier -- i.e. the Laplacian matrix of a weighted subgraph of $G$.

\subsection{Leverage Scores and Row Sampling}
\label{Leverage Scores and Row Sampling}
For any $\bv{B} \in \mathbb{R}^{m \times n}$ with rank $r$, consider the reduced singular value decomposition, $\bv{B}=\bv{U}\bv{\Sigma}\bv{V}^{\top}$. $\bv{U} \in \mathbb{R}^{m\times r}$ and $\bv{V} \in \mathbb{R}^{n \times r}$ have orthonormal columns and $\bv{\Sigma} \in \mathbb{R}^{r \times r}$ is diagonal and contains the non-zero singular values of $\bv{B}$. Then, $\bv{B}^\top\bv{B} = \bv{V}\bv{\Sigma}\bv{U}^\top\bv{U}\bv{\Sigma}\bv{V}^{\top} = \bv{V}\bv{\Sigma}^2\bv{V}^{\top}$. We let $\bv{K}^+$ denote the Moore-Penrose pseudoinverse of $\bv{K} = \bv{B}^\top \bv{B}$:
\begin{align*}
\bv{K}^+ = \bv{V}(\bv{\Sigma}^{-1})^2\bv{V}^{\top}.
\end{align*}
 The leverage score, $\tau_i$, for a row $\bv{b}_i$ in $\bv{B}$ is defined as
\begin{align*}
\tau_i &\eqdef \bv{b}_i^\top \bv{K}^+ \bv{b}_i = \bv{u}_i^\top \bv{\Sigma V^\top} (\bv{V} \bv{\Sigma}^{-2} \bv{V^\top) V \Sigma} \bv{u}_i = \norm{\bv{u}_i}_2^2 \le 1.
\end{align*}
The last inequality follows from the fact that every row in a matrix with orthonormal columns has norm less than 1. In a graph, $\tau_i = r_i w_i$, where $r_i$ is the \emph{effective resistance} of edge $i$ and $w_i$ is the edge's weight. Furthermore,
\begin{align*}
\sum_{i=1}^m \tau_i = \tr(\bv{B}\bv{K^+}\bv{B^\top}) = \norm{\bv{U}}_F^2 = r = \rank(\bv{B}).
\end{align*}


It is well known that by sampling the rows of $\bv{B}$ according to their leverage scores it is possible to obtain a matrix $\bv{\tilde B}$ such that $\bv{\tilde K} = \bv{\tilde B}^\top \bv{\tilde B} \approx_\epsilon \bv{K}$ with high probability. Furthermore, if obtaining exact leverage scores is computationally difficult, it suffices to sample by upper bounds on the scores. Typically, rows are sampled \emph{with replacement} with probability proportional to their leverage score \cite{graphSparsificationEffectiveResistance,pengV2}. We require an alternative procedure for sampling edges \emph{independently}.

%

\begin{lemma}[Spectral Approximation via Leverage Score Sampling]\label{sparsifier_sampling}
Let $\bs{\tilde \tau}$ be a vector of leverage score overestimates for $\bv{B}$'s rows such that $\bv{\tilde \tau}_i \ge \bv{\tau}_i$ for all $i \in [m]$. For $0 < \epsilon < 1$ and fixed constant $c$, define the sampling probability for row $\bv{b}_i$ to be $p_i = \min \{1,c \log n \epsilon^{-2} \tilde{\tau}_i\}$. Define a diagonal sampling matrix $\bv{W}$ with $\bv{W}(i,i) = \frac{1}{p_i}$ with probability $p_i$ and $\bv{W}(i,i) = 0$ otherwise. With high probability,
\begin{align*}
\bv{\tilde K} = \bv{B}^\top \bv{W} \bv{B} \approx_\epsilon \bv{K}.
\end{align*}
Furthermore, $\bv{W}$ has $O(\norm{\bs{\tilde \tau}}_1\log n \epsilon^{-2})$ non-zeros with high probability. \end{lemma}
A proof of Lemma \ref{sparsifier_sampling} based on a matrix concentration result from \cite{tropp2012user} can be found in \cite{uniformSampling} (Lemma 4). Note that, when applied to the vertex edge incidence matrix of a graph, leverage score sampling is equivalent to effective resistance sampling, as introduced in \cite{graphSparsificationEffectiveResistance} for graph sparsification.

\subsection{Sparse Recovery}
\label{sparse recovery}
While we cannot sample by leverage score directly in the streaming model, we can use a sparse recovery primitive to sample edges from a set of linear sketches. We use an $\ell_2$ heavy hitters algorithm that, for any vector $\bv{x}$, lets us recover from a small linear sketch $\bs{\Phi} \bv{x}$, the index $i$ and the approximate value of $\bv{x}_i$ for all $i$ such that $\bv{x}_i > \eta ||\bv{x}||_2$.

\begin{lemma}[$\ell_2$ Heavy Hitters]
\label{sparse_recovery_primitive}
For any $\eta >0$, there is a decoding algorithm $D$ and a distribution on matrices $\bv{\Phi}$ in $\R^{O(\eta^{-2} \plog( N)) \times N}$ such that, for any $\bv{x}\in \R^N$, given $\bs{\Phi} \bv{x}$, the algorithm $D$ returns a vector $\bv{w}$ such that
$\bv{w}$ has $O(\eta^{-2}\plog( N))$ non-zeros and satisfies
$$ ||\bv{x}-\bv{w}||_\infty \leq \eta ||\bv{x}||_2.$$
with probability $1-N^{-c}$ over the choice of $\bs{\Phi}$.
The sketch $\bv{\Phi} \bv{x}$ can be maintained and decoded in $O(\eta^{-2}\plog( N))$ space. 
\end{lemma}
This procedure allows us to distinguish from a sketch whether or not a specified entry in $\bv{x}$ is equal to 0 or has value $> 2\eta\|\bv{x}\|_2$. 
We give a proof of Lemma \ref{sparse_recovery_primitive} in Appendix \ref{sparse_recovery_appendix}

\section{Algorithm Overview}\label{algorithm_overview}
Before formally presenting a proof of our main result, Theorem \ref{main_sparsification_theorem}, we give an informal overview of the algorithm to provide intuition. 

\subsection{Effective Resistances}
As explained in Section \ref{Leverage Scores and Row Sampling}, spectral sparsifiers can be generated by sampling edges, i.e. rows of the vertex edge incidence matrix. For an unweighted graph $G$, each edge $e$ is sampled independently with probability proportional to its leverage score, $\tau_e$. After sampling, we reweight and combine any sampled edges. The result is a subgraph of $G$ containing, with high probability, $O(\frac1{\epsilon^2}n\log n)$ edges and spectrally approximating $G$.

If we view $G$ as an electrical circuit, with each edge representing a unit resistor, the leverage score of an edge $e=(i,j)$ is equivalent  to its effective resistance. This value can be computed by forcing $1$ unit of current out of vertex $i$ and $1$ unit of current into vertex $j$. The resulting voltage difference between the two vertices is the effective resistance of $e$. Qualitatively, if the voltage drop is low, there are many low resistance (i.e. short) paths between $i$ and $j$. Thus, maintaining a direct connection between these vertices is less critical in approximating $G$, so $e$ is less likely to be sampled.
Effective resistance can be computed as:
\begin{align*}
\tau_e = \bv{b}_e^\top \bv{K}^+\bv{b}_e
\end{align*}

Note that $\tau_e$ can be computed for any pair of vertices, $(i,j)$, or in other words, for any possible edge in $G$. We can evaluate  $\bv{b}_e^\top \bv{K}^+\bv{b}_e$ even if $e$ is not present in the graph. Thus, we can reframe our sampling procedure. Instead of just sampling edges actually in $G$, imagine we run a sampling procedure for \emph{every possible} $e$. When recombining edges to form a spectral sparsifier, we separately check whether each edge $e$ is in $G$ and only insert into the sparsifier if it is. 
\subsection{Sampling in the Streaming Model}
With this procedure in mind, a sampling method that works in the streaming setting requires two components. First, we need to obtain a constant factor approximation to $\tau_e$ for any $e$. Known sampling algorithms, including our Lemma \ref{sparsifier_sampling}, are robust to this level of estimation. Second, we need to compress our edge insertions and deletions in such a way that, during post-processing of our sketch, we can determine whether or not a sampled edge $e$ actually exists in $G$. 

The first requirement is achieved through the recursive procedure given in \cite{pengV1}. We will give the overview shortly but, for now, assume that we have access to a coarse sparsifier, $\bv{\tilde K} \approx_{1/2} \bv{K}$. Computing $\bv{b}_e^\top \bv{\tilde K}^+\bv{b}_e$ gives a 2 factor multiplicative  approximation of $\tau_e$ for each $e$. Furthermore, as long as $\bv{\tilde K}$ has sparsity $O(n\log n)$, the computation can be done in small space using an iterative system solver (e.g. conjugate gradient) or a nearly linear time solver for symmetric diagonally dominant matrices (e.g. \cite{koutis2011nearly}).

Solving part two (determining which edges are actually in $G$) is a bit more involved.
As a first step, consider writing
\begin{align*}
\tau_e = \bv{b}_e^\top \bv{K}^+\bv{K}\bv{K}^+\bv{b}_e = \norm{\bv{B}\bv{K}^+\bv{b}_e}_2^2 = \norm{\bv{S}\bv{B}_n\bv{K}^+\bv{b}_e}_2^2.
\end{align*}
Referring to Section \ref{notation}, recall that $\bv{B} = \bv{S}\bv{B}_n$ is exactly the same as a standard vertex edge incidence matrix except that rows in $\bv{B}_n$ corresponding to nonexistent edges are zeroed out instead of removed. Denote $\bv{x}_e = \bv{S}\bv{B}_n\bv{K}^+\bv{b}_e$. Each nonzero entry in $\bv{x}_e$ contains the voltage difference across some edge (resistor) in $G$ when one unit of current is forced from $i$ to $j$. 

When $e$ is not in $G$, then the $e^\text{th}$ entry of $\bv{x}_e$,  $\bv{x}_e(e)$ is $0$. If $e$ is in $G$, $\bv{x}_e(e) = \tau_e$. Furthermore, $\norm{\bv{x}_e}_2^2 = \tau_e$. 
Given a space allowance of $\plog(n)$, the sparse recovery algorithm from Lemma \ref{sparse_recovery_primitive} allows us to recover an entry if it accounts for at least an $\Omega(1/\plog(n))$ fraction of the total $\ell_2$ norm. Currently, $\bv{x}_e(e)/\norm{\bv{x}_e}_2 = \sqrt{\tau_e}$, which could be much smaller than $O(1/\plog(n))$. However, suppose we had a sketch of $\bv{x}_e$ with all but a $\tau_e$ fraction of edges randomly sampled out. Then, we would expect $\norm{\bv{x}_e}_2^2 \approx \tau_e^2$ and
thus, $\bv{x}_e(e)/\norm{\bv{x}_e}_2 = O(1) = \Omega(1/\plog(n))$ and sparse recovery would successfully indicate whether or not $e\in G$. What's more, randomly zeroing out entries of $\bv{x}_e$ can serve as our main sampling routine for edge $e$. This process will set $\bv{x}_e(e)=0$ with probability $(1-\tau_e)$, exactly what we wanted to sample by in the first place!

However, how do we go about sketching every appropriately sampled $\bv{x}_e$? Well, consider subsampling our graph at geometrically decreasing rates, $1/2^s$ for $s \in \{0,1,...O(\log n) \}$. Maintain linear sketches $\bv{\Pi}_1 \bv{B}_1,...\bv{\Pi}_{O(\log n)}\bv{B}_{O(\log n)}$ of the vertex edge incidence matrix for every subsampled graph using the $\ell_2$ sparse recovery sketch distribution from Lemma \ref{sparse_recovery_primitive}.
When asked to output a spectral sparsifier, for every possible edge $e$, we compute using $\bv{\tilde K}$ a rate $1/2^s$ that approximates $\tau_e$. 

Since each sketch is linear, 
we can just multiply $\bv{\Pi}_{1/2^s} \bv{B}_{1/2^s}$ on the right by $\bv{\tilde K}^+\bv{b}_e$ to compute
\begin{align*}
\bv{\Pi}_{1/2^s} \bv{B}_{1/2^s}\bv{\tilde K}^+\bv{b}_e \approx \bv{\Pi}_{1/2^s} \bv{x}_e^{1/2^s},
\end{align*}
where $\bv{x}_e^{1/2^s}(e)$ is $\bv{x}_e$ sampled at rate $1/2^s \approx \tau_e$. 
Then, as explained, we can use our sparse recovery routine to determine whether or not $e$ is present. If it is, we have obtained a sample for our spectral sparsifier!

\subsection{A Chain of Coarse Sparsifiers}
The final required component is access to some sparse $\bv{\tilde K} \approx_{1/2} \bv{K}$. This coarse sparsifier is obtained recursively by constructing a chain of matrices, $\begin{bmatrix}\bv{K}(0), \bv{K}(1), \ldots, \bv{K}(d), \bv{K}\end{bmatrix}$ each weakly approximating the next. Specifically, imagine producing $\bv{K}(d)$ by adding a fairly light identity matrix to $\bv{K}$. As long as the identity's weight is small compared to $\bv{K}$'s spectrum, $\bv{K}(d)$ approximates $\bv{K}$. Add even more weight to the diagonal to form $\bv{K}(d-1)$. Again, as long as the increase is small, $\bv{K}(d-1)$ approximates $\bv{K}(d)$. We continue down the chain until $\bv{K}(0)$, which will actually have a heavy diagonal after all the incremental increases. Thus, $\bv{K}(0)$ can be approximated by an appropriately scaled identity matrix, which is clearly sparse. Miller and Peng show that parameters can be chosen such that $d = O(\log n)$ \cite{pengV1}.

Putting everything together, we maintain $O(\log n)$ sketches for $\begin{bmatrix}\bv{K}(0), \bv{K}(1), \ldots, \bv{K}(d), \bv{K}\end{bmatrix}$. We first use a weighted identity matrix as a coarse approximation for $\bv{K}(0)$, which allows us to recover a good approximation to $\bv{K}(0)$ from our sketch. This approximation will in turn be a coarse approximation for $\bv{K}(1)$, so we can recover a good sparsifier of $\bv{K}(1)$. Continuing up the chain, we eventually recover a good sparsifier for our final matrix, $\bv{K}$. 

\section{Recursive Sparsifier Construction}\label{recursive_algorithm}
In this section, we formalize a recursive procedure for obtaining a chain of coarse sparsifiers that was introduced by Miller and Peng  -- ``Introduction and Removal of Artificial Bases'' \cite{pengV1}. We prove Theorem \ref{main_sparsification_theorem} by combining this technique with the sampling algorithm developed in Section \ref{streaming_row_sampling}.


\begin{theorem}[Recursive Sparsification -- \cite{pengV1}, Section 4]
\label{miller_peng_chain}
Consider any PSD matrix $\bv{K}$ with maximum eigenvalue bounded from above by $\lambda_{u}$ and minimum non-zero eigenvalue bounded from below by $\lambda_{l}$. Let $d = \lceil \log_2 (\lambda_{u}/\lambda_{l})\rceil$.  For $\ell \in \{0,1, 2, ... , d\}$, define
\begin{align*}
\gamma(\ell) = \frac{\lambda_{u}}{2^\ell}.
\end{align*}
So, $\gamma(0) = \lambda_u$ and $\gamma(d) \leq \lambda_{l}$. Then the chain of PSD matrices, $\begin{bmatrix}\bv{K}(0), \bv{K}(1), \ldots, \bv{K}(d)\end{bmatrix}$ with
\begin{align*}
\bv{K}(\ell) = \bv{K} + \gamma(\ell)\bv{I}_{n\times n},
\end{align*}
satisfies the following relations:
\begin{enumerate}
  \item $\bv{K} \preceq_r \bv{K}(d) \preceq_r 2\bv{K}$,
  \item $\bv{K}(\ell) \preceq \bv{K}(\ell-1) \preceq 2\bv{K}(\ell)$ for all $\ell \in \{1,\ldots, d\}$,
  \item $\bv{K}(0) \preceq 2\gamma(0)\bv{I} \preceq 2\bv{K}(0)$.
\end{enumerate}
When $\bv{K}$ is the Laplacian of an unweighted graph, its largest eigenvalue $\lambda_{max} < 2n$ and its smallest non-zero eigenvalue $\lambda_{min} > 8/n^2$. Thus the length of our chain, $d = \lceil \log_2 \lambda_{u}/\lambda_{l}\rceil$, is $O(\log n)$.
\end{theorem}

For completeness, we include a proof of Theorem \ref{miller_peng_chain} in Appendix \ref{miller_peng_appendix}.  Now, to prove our main result, we need to state the sampling primitive for streams that we develop in Section \ref{streaming_row_sampling}. This procedure maintains a linear sketch of a vertex edge incidence matrix $\bv{B}$, and using a coarse sparsifier of $\bv{K}(\ell) = \bv{B}^\top \bv{B} + \gamma(\ell) \bv{I}$, performs independent edge sampling as required by Lemma \ref{sparsifier_sampling}, to obtain a better sparsifier of $\bv{K}(\ell)$.




\begin{theorem}\label{refinement}
Let $\bv{B} \in \mathbb{R}^{n \times m}$ be the vertex edge incidence matrix of an unweighted graph $G$, specified by an insertion-deletion graph stream. Let $\gamma = \poly(n)$ be a fixed parameter and consider $\bv{K} = \bv{B}^\top \bv{B} + \gamma \bv{I}$. For any $0< \epsilon < 1$, there exists a sketching procedure $\maintain(\bv{B},\epsilon)$ that outputs an $O(n \plog (n))$ sized sketch $\bv{\Pi}\bv{B}$. There exists a corresponding recovery algorithm \texttt{RefineSparsifier} running in $O(n \plog (n))$ space, such that,
if $\bv{\tilde K}$ is a spectral approximation to $\bv{K}$ with $O(n \log n)$ non-zeros and $c\bv{K} \preceq_{r} \bv{\tilde K} \preceq_r \bv{K}$ for some constant $0 < c < 1$ then:\\

 $\texttt{RefineSparsifier}(\bv{\Pi}\bv{B}, \bv{\tilde K}, \gamma, \epsilon,c)$ returns, with high probability, $\bv{\tilde K}_\epsilon = \bv{\tilde B}_\epsilon^\top \bv{\tilde B}_\epsilon + \gamma{\bv{I}}$, where $(1-\epsilon)\bv{K} \preceq_{r} \bv{\tilde K}_\epsilon \preceq_r (1+\epsilon)\bv{K}$, and $\bv{\tilde B}_\epsilon$ contains only $O(\epsilon^{-2} c^{-1} n\log n)$ reweighted rows of $\bv{B}$ with high probability.
$\texttt{RefineSparsifier}$ runs in $O(n^2\plog (n))$ time.
\end{theorem}

Using this sampling procedure, we can initially set $\bv{\tilde K} =2\gamma(0)\bv{I}$ and use it obtain a sparsifier for $\bv{K}(0)$ from a linear sketch of $\bv{B}$. This sparsifier is then used on a second sketch of $\bv{B}$ to obtain a sparsifier for $\bv{K}(1)$, and so on. Working up the chain, we eventually obtain a sparsifier for our original $\bv{K}$. While sparsifier recovery proceeds in several levels, we construct all required sketches in a \emph{single pass} over edge insertions and deletions. Recovery is performed in post-processing.

\begin{proof}[Proof of Theorem \ref{main_sparsification_theorem}]
Let $\bv{K}$ be the Laplacian of our graph $G$. Process all edge insertions and deletions, using $\maintain$ to produce a sketch, $(\bv{\Pi}\bv{B})_{\ell}$ for each $\ell \in \{0, 1,\ldots,  \lceil \log_2 \lambda_{u}/\lambda_{l}\rceil + 1\}$.
We then use Theorem \ref{refinement} to recover an $\epsilon$ approximation, $\bv{\tilde K}(\ell)$, for any  $\bv{K}(\ell)$ given an $\epsilon$ approximation for $\bv{K}(\ell-1)$. First, consider the base case, $\bv{K}(0)$. Let:
\begin{align*}
\bv{\tilde K}(0) = \texttt{RefineSparsifier}((\bv{\Pi}\bv{B})_{0}, \gamma(0)\bv{I}, \gamma(0), \epsilon,\frac{1}{2}).
\end{align*}
By Theorem \ref{miller_peng_chain}, Relation 3:
\begin{align*}
\frac{1}{2}\bv{K}(0) \preceq \gamma(0)\bv{I} \preceq \bv{K}(0).
\end{align*}
Thus, with high probability, $(1-\epsilon)\bv{K}(0) \preceq_{r} \bv{\tilde K}(0) \preceq_r (1+\epsilon)\bv{K}(0)$ and $\bv{\tilde K}(0)$ contains $O((1/2)^{-1} \cdot n\log n \cdot \epsilon^{-2}) = O(\epsilon^{-2} n\log n)$ entries.

Now, consider the inductive case. Suppose we have some $\bv{\tilde K}(\ell - 1)$ such that $(1-\epsilon)\bv{K}(\ell - 1) \preceq_{r} \bv{\tilde K}(\ell - 1) \preceq_r (1+\epsilon)\bv{K}(\ell - 1)$. Let:
\begin{align*}
\bv{\tilde K}(\ell) = \texttt{RefineSparsifier}((\bv{\Pi}\bv{B})_{\ell}, \frac{1}{2(1+\epsilon)}\bv{\tilde K}(\ell - 1), \gamma(\ell), \epsilon, \frac{1-\epsilon}{2(1+\epsilon)}).
\end{align*}
By Theorem \ref{miller_peng_chain}, Relation 2:
\begin{align*}
\frac{1}{2}\bv{K}(\ell) \preceq \frac{1}{2}\bv{K}(\ell-1) \preceq \bv{K}(\ell).
\end{align*}
Furthermore, by assumption we have the inequalities:
\begin{align*}
\frac{1-\epsilon}{1+\epsilon}\bv{K}(\ell - 1) \preceq_{r} \frac{1}{1+\epsilon}\bv{\tilde K}(\ell - 1) \preceq_r \bv{K}(\ell - 1).
\end{align*}
Thus:
\begin{align*}
\frac{1-\epsilon}{2(1+\epsilon)}\bv{K}(\ell) \preceq_{r} \frac{1}{2(1+\epsilon)}\bv{\tilde K}(\ell - 1) \preceq_r \bv{K}(\ell).
\end{align*}
So, with high probability $\texttt{RefineSparsifier}$ returns $\bv{\tilde K}(\ell)$ such that $(1-\epsilon)\bv{K}(\ell) \preceq_{r} \bv{\tilde K}(\ell) \preceq_r (1+\epsilon)\bv{K}(\ell)$ and $\bv{\tilde K}(\ell)$ contains just $O((\frac{2(1+\epsilon)}{1-\epsilon})^2 \epsilon^{-2} n\log n ) = O(\epsilon^{-2} n\log n)$ nonzero elements. It is important to note that there is no ``compounding of error'' in this process. Every $\bv{\tilde K}(\ell)$ is an $\epsilon$ approximation for $\bv{K}(\ell)$. Error from using $\bv{\tilde K}(\ell-1)$ instead of $\bv{K}(\ell-1)$ is absorbed by a constant factor increase in the number of rows sampled from $\bv{B}$. The corresponding increase in sparsity for $\bv{K}(\ell)$ does not compound -- in fact Theorem \ref{refinement} is completely agnostic to the sparsity of the coarse approximation $\bv{\tilde K}$ used.

Finally, to obtain a bonafide graph sparsifier (a weighted subgraph of our streamed graph), let:
\begin{align*}
\bv{\tilde K} = \texttt{RefineSparsifier}((\bv{\Pi}\bv{B})_{d + 1}, \frac{1}{2(1+\epsilon)}\bv{\tilde K}(d), 0, \epsilon, \frac{1-\epsilon}{2(1+\epsilon)}).
\end{align*}
As in the inductive case,
\begin{align*}
\frac{1-\epsilon}{2(1+\epsilon)}\bv{K} \preceq_{r} \frac{1}{2(1+\epsilon)}\bv{\tilde K}(d) \preceq_r \bv{K}.
\end{align*}
Thus, it follows that, with high probability, $\bv{\tilde K}$ has sparsity $O( \epsilon^{-2} n\log n)$ and $(1-\epsilon)\bv{K} \preceq_{r} \bv{\tilde K} \preceq_r (1+\epsilon)\bv{K}$. Since we set $\gamma$ to 0 for this final step, $\bv{\tilde K}$ simply equals $\bv{\tilde B}^\top \bv{\tilde B}$ for some $\bv{\tilde B}$ that contains reweighted rows of $\bv{B}$. Any vector in the kernel of $\bv{B}$ is in the kernel of $\bv{\tilde B}$, and thus any vector in the kernel of $\bv{K}$ is in the kernel of $\bv{\tilde K}$. Thus, we can strengthen our approximation to:
\begin{align*}
(1-\epsilon)\bv{K} \preceq \bv{\tilde K} \preceq (1+\epsilon)\bv{K}.
\end{align*}
We conclude that $\bv{\tilde K}$ is the Laplacian of some graph $H$ containing $O(\epsilon^{-2} n\log n)$ reweighted edges and approximating $G$ spectrally to precision $\epsilon$. Finally, note that we require $d+1 = O(\log n)$ recovery steps, each running in $O(n^2\plog (n))$ time. Thus, our total recovery time is $O(n^2\plog (n))$.
\end{proof}

\section{Streaming Row Sampling}\label{streaming_row_sampling}


In this section, we develop the sparsifier refinement routine required for Theorem \ref{main_sparsification_theorem}.

\begin{proof}[Proof of Theorem \ref{refinement}]

Outside of the streaming model, given full access to $\bv{B}$ rather than just a sketch $\bv{\Pi}\bv{B}$ it is easy to implement $\refine$ via leverage score sampling. Letting $\oplus$ denote appending the rows of one matrix to another, we can define $\bv{B}_\gamma = \bv{B} \oplus \sqrt{\gamma(\ell)} \cdot \bv{I}$, so $\bv{K} = \bv{B}^\top \bv{B} + \gamma \bv{I} = \bv{B}_\gamma^\top \bv{B}_\gamma$. Since $\tau_i = \bv{b}_i^\top  \bv{ K}^+ \bv{b}_i$ and $c\bv{K} \preceq_r \bv{\tilde K} \preceq_r \bv{K}$, for any row of $\bv{B}_\gamma$ we have
\begin{align*}
 \tau_i \le \bv{b}_i^\top \bv{\tilde K}^+ \bv{b}_i \le \frac{1}{c}\tau_i .
\end{align*}

Let $\tilde{\tau_i} = \bv{b}_i^\top \bv{\tilde K}^+ \bv{b}_i$ be the leverage score of $\bv{b}_i$ approximated using $\bv{\tilde K}$. Let $\bs{\tilde \tau}$ be the vector of approximate leverage scores, with the leverage scores of the $n$ rows corresponding to $\sqrt{\gamma(\ell)} \bv{I}$ rounded up to $1$. While not strictly necessary, including rows of the identity with probability $1$ will simplify our analysis in the streaming setting. Using this $\bs{\tilde \tau}$ in Lemma \ref{sparsifier_sampling}, we can obtain $\bv{\tilde K}_\epsilon \approx_\epsilon \bv{K}$ with high probability. Since $\norm{\bs{\tilde \tau}}_1 \le \frac{1}{c}\norm{\bs{\tau}}_1 + n \le \frac{1}{c} \cdot \rank(\bv{B}) + n \le \frac{n}{c} + n$, we can write  $\bv{\tilde K}_\epsilon = \bv{\tilde B}_\epsilon^\top \bv{\tilde B}_\epsilon + \gamma \bv{I}$, where $\bv{\tilde B}_\epsilon$ contains $O(\epsilon^{-2}c^{-1} n\log n)$ reweighted rows of $\bv{B}$ with high probability.

The challenge in the semi-streaming setting is actually sampling edges given only a sketch of $\bv{B}$. The general idea is explained in Section \ref{algorithm_overview}, with detailed pseudocode included below. 
\begin{framed}{\noindent\bfseries Streaming Sparsifier Refinement}

\paragraph{$\maintain(\bv{B}, \epsilon)$:}
\begin{enumerate}

\item For $s \in \{1,...O(\log n)\}$ let $h_s: {n \choose 2} \rightarrow \{0,1 \}$ be a uniform hash function. Let $\bv{B}_s$ be $\bv{B}$ with all rows except those with $\prod_{j \le s} h_j(e) = 0$ zeroed out. So $\bv{B}_s$ is $\bv{B}$ with rows sampled independently at rate $\frac{1}{2^s}$. $\bv{B}_0$ is simply $\bv{B}$.  

\item Maintain sketchs $\bv{\Pi}_0\bv{B}_0, \bv{\Pi}_{1}\bv{B}_{1}, ... , \bv{\Pi}_{O(\log n)}\bv{B}_{O(\log n)}$ where $\{\bv{\Pi}_0,\bv{\Pi}_1,...\bv{\Pi}_{O(\log n)} \}$ are drawn from the distribution from Lemma \ref{sparse_recovery_primitive} with $\eta = \frac{\epsilon}{c_1\sqrt{\log n}}$.
\item Output all of these sketches stacked: $\bs{\Pi}\bv{B} = \bv{\Pi}_0\bv{B}_0 \oplus \ldots \oplus \bv{\Pi}_{O(\log n)}\bv{B}_{O(\log n)}$.

\end{enumerate}

\paragraph{$\refine(\bv{\Pi}\bv{B}, \bv{\tilde K}, \gamma, \epsilon,c)$:}
\begin{enumerate}
\item Compute $\bv{\Pi}_s\bv{B}_s\bv{\tilde K}^+$ for each $s \in \{0,1,2,...O(\log n)\}$.

\item For every edge $e$ in the set of ${n \choose 2}$ possible edges:
\begin{enumerate}
\item Compute $\tilde \tau_e = \bv{b}_e^\top \bv{\tilde K}^+ \bv{b}_e$ and $p_e =  c_2\tilde \tau_e \log n \epsilon^{-2}$, where $c_2$ is the oversampling constant from Lemma \ref{sparsifier_sampling}. Choose $s$ such that $ \min \{1, p_e \} \le \frac{1}{2^s} \le \min \{1, 2 p_e \}$.

\item Compute $ \bv{\Pi}_s\bv{x}_e = \bv{\Pi}_s\bv{B}_s\bv{\tilde K}^+\bv{b}_e$ and run the heavy hitters algorithm of Lemma \ref{sparse_recovery_primitive}. Determine whether or not $\bv{x}_e = 0$ or $\bv{x}_e \geq \frac{\epsilon}{c_1\sqrt{\log n}}\|\bv{B}_s\bv{\tilde K}^+\bv{b}_e\|_2$ by checking whether the returned $\bv{w}_e > \tilde \tau_e/2$. 

\item If it is determined that $\bv{ x}_e(e) \neq 0$ set $\bv{W}(e,e) = 2^s$.

\end{enumerate}

\item Output $\bv{\tilde K}_\epsilon = \bv{B}^\top \bv{W} \bv{B} + \gamma \bv{I}$.
\end{enumerate}
\end{framed}

We show that every required computation can be performed in the dynamic semi-streaming model and then prove the correctness of the sampling procedure. 
\subsubsection*{Implementation in the Semi-Streaming Model.}
Assuming access to uniform hash functions, $\maintain$ requires $O(n\plog (n))$ space in total and can be implemented in the dynamic streaming model. When an edge insertion comes in, use $\{ h_s \}$ to compute which $\bv{B}_s$'s should contain the inserted edge, and update the corresponding sketches. For an edge deletion, simply update the sketches to add $-\bv{b}_e$ to each appropriate $\bv{B}_s$.

Unfortunately, storing $O(\log n)$ uniform hash functions over ${n\choose 2}$ requires $O(n^2\log n)$ space, and is thus impossible in the semi-streaming setting. If Section \ref{pseudorandomness} we show how to cope with this issue by using a small-seed pseudorandom number generator.

Step 1 of $\refine$ can also be implemented in $O(n \plog n)$ space. Since $\bv{\tilde K}$ has $O(n\log n)$ non-zeros and $\bv{\Pi}_s\bv{B}_s$ has $O(\plog n)$ rows, computing $\bv{\Pi}_s\bv{B}_s\bv{\tilde K}^+$ requires $O(\plog n)$ linear system solves in $\bv{\tilde K}$. We can use an iterative algorithm or a nearly linear time solver for symmetric diagonally dominant matrices to find solutions in $O(n \plog n)$ space total.

For step 2(a), the $s$ chosen to guarantee $ \min \{1, p_e \} \le \frac{1}{2^s} \le \min \{1, 2p_e \}$ could in theory be larger than the index of the last sketch $\bs{\Pi}_i\bv{B}_i$ maintained. However, if we take $O(\log n)$ samplings, our last will be empty with high probability. Accordingly, all samplings for higher values of $s$ can be considered empty as well and we can just skip steps 2(b) and 2(c) for such values of $s$. Thus, $O(\log n)$ sampling levels are sufficient.

Finally, by our requirement that $\bv{\tilde K}$ is able to compute $\frac{1}{c}$ factor leverage score approximations, with high probability, Step 2 samples at most $O(n \log n \epsilon^{-2})$ edges in total (in addition to selecting $n$ identity edges). Thus, the procedure's output can be stored in small space.

\subsubsection*{Correctness}

To apply our sampling lemma, we need to show that, with high probability, $\refine$ independently samples each row of $\bv{B}$ with probability $\hat p_e$ where $ \min \{1, p_e \} \le \hat p_e \le \min \{1, 2p_e \}$. Since the algorithm samples the rows of $\sqrt{\gamma} \bv{I}$ with probability $1$, and since $\tau_e \le \tilde \tau_e \le \frac{1}{c} \tau_e$ for all $e$, by Lemma \ref{sparsifier_sampling}, with high probability, $\bv{\tilde K}_\epsilon = \bv{B} \bv{W} \bv{B} + \gamma \bv{I} = \bv{\tilde B}_\epsilon^\top \bv{\tilde B}_\epsilon + \gamma \bv{I}$ is a $(1\pm \epsilon)$ spectral sparsifier for $\bv{K}$. Furthermore, $\bv{\tilde B}_\epsilon$ contains $O(\epsilon^{-2}c^{-1} n\log n )$ reweighted rows of $\bv{B}$.

In $\refine$, an edge is only included in $\bv{\tilde K_\epsilon}$ if it is included in the $\bv{B}_{s(e)}$ where
\begin{align*}
\min \{1, p_e \}\le \frac{1}{2^{s(e)}} \le \min \{1, 2p_e \}.
\end{align*}

The probability that $\bv{b}_e$ is included in the sampled matrix $\bv{B}_{s(e)}$ is simply $1/2^{s(e)}$, and sampling is done independently using uniform hash functions. So, we just need to show that, with high probability, any $\bv{b}_e$ included in its respective $\bv{B}_{s(e)}$ is recovered by Step 2(b).

Let $\bv{x}_e = \bv{B}\bv{\tilde K}^+\bv{b}_e$ and $\bv{x}_e^{s(e)} = \bv{B}_{s(e)}\bv{\tilde K}^+\bv{b}_e$. 
As explained in Section \ref{algorithm_overview},
\begin{align}
\label{entry_is_lev_score}
\bv{x}_e^{s(e)}(e) = \bv{x}_e(e) &= \bv{1}_e \bv{B} \bv{\tilde K}^+ \bv{b}_e = \bv{b}_e^\top \bv{\tilde K}^+ \bv{b}_e = \tilde \tau_e
\end{align} 
Furthermore, we can compute:
\begin{align}
\norm{\bv{x}_e}_2^2 &=  \bv{b}_e^\top  \bv{\tilde K}^+ \bv{B}^\top \bv{B} \bv{\tilde K}^+ \bv{b}_e\nonumber\\
&\le  \bv{b}_e^\top  \bv{\tilde K}^+ \bv{B}_\gamma^\top \bv{B}_\gamma \bv{\tilde K}^+ \bv{b}_e\tag{Since $\bv{B}^\top \bv{B} \preceq \bv{B}_\gamma^\top \bv{B}_\gamma$}\\
&\le \frac{1}{c} \cdot \bv{b}_e^\top \bv{\tilde K}^+ \bv{b}_e\tag{Since $c \left (\bv{B}_\gamma^\top \bv{B}_\gamma \right ) \preceq \bv{\tilde K}$}\\
\label{real_eqn}
&= \frac{1}{c} \tilde \tau_e
\end{align} 
Now, writing $\hat p_e = \frac{1}{2^{s(e)}}$, we \emph{expect} $\|\bv{x}_e^{s(e)}(e)\|_2^2$ to equal $\hat{p}_e \norm{\bv{x}_e}_2^2 = O(\tilde{\tau}_e^2\log n\epsilon^{-2})$. We want to argue that the norm falls close to this value with high probability. This follows from claiming that no entry in $\bv{x}_e$ is too large.
For any edge $e' \neq e$ define:
\begin{align*}
\tilde \tau_{e',e} \eqdef \bv{x}_e(e') &= \bv{1}_{e'} \bv{B} \bv{\tilde K}^+ \bv{b}_e = \bv{b}_{e'}^\top \bv{\tilde K}^+ \bv{b}_e.
\end{align*}

\begin{lemma}\label{leverage_score_bound}$\tilde \tau_{e',e} \le \tilde \tau_e$\end{lemma}.
\begin{proof}


Consider $\bv{\tilde  v}_e = \bv{\tilde K}^+ \bv{b}_e$. 
Let $e = (u_1,u_2)$ and $e' = (u_1',u_2')$. If we have $|\bv{\tilde v}_e(u_1') - \bv{\tilde v}_e(u_2)'| \le |\bv{\tilde v}_e(u_1) - \bv{\tilde v}_e(u_2)|$ then

\begin{align*}
\bv{b}_{e'}^\top \bv{ \tilde v}_e  = \bv{b}_{e'}^\top \bv{ \tilde K}^+ \bv{b}_e &\le \bv{b}_{e}^\top \bv{ \tilde K}^+ \bv{b}_e = \bv{b}_{e}^\top \bv{ \tilde v}_e,
\end{align*}
which implies $\tilde \tau_{e',e} \le  \tilde \tau_{e}$ as desired.

Now, $\bv{\tilde K}$ is a weighted graph Laplacian added to a weighted identity matrix. Thus it is full rank and diagonally dominant. Since it has full rank, $\bv{\tilde K} \bv{\tilde v}_e = \bv{\tilde K} \bv{\tilde K}^+ \bv{b}_e = \bv{b}_e$. Since $\bv{\tilde K}$ is diagonally dominant and since $\bv{b}_e$ is zero everywhere except at $\bv{b}_e(u_1) = 1$ and $\bv{b}_e(u_2) = -1$, it must be that  $\bv{\tilde v}_e(u_1)$ is the maximum value of $\bv{\tilde v}_e$ and $\bv{\tilde v}_e(u_2)$ is the minimum value. So $|\bv{\tilde v}_e(u_1') - \bv{\tilde v}_e(u_2)'| \le |\bv{ \tilde v}_e(u_1) - \bv{\tilde v}_e(u_2)|$ and $\tilde \tau_{e',e} \le \tilde \tau_{e}$.

\end{proof}
From Lemma \ref{leverage_score_bound}, the vector $\frac{1}{\tilde \tau_e}\bv{x}_e$ has all entries (and thus all squared entries) in $[0,1]$ so we can apply a Chernoff/Hoeffding bound to show concentration for $\|\frac{1}{\tilde \tau_e}\bv{x}_e^{s(e)}\|_2^2$. Specifically, we use the standard multiplicative bound \cite{Hoeffding:1963}:
\begin{align}
\Pr(X > (1+\delta)\E{X}) < e^{- 2\delta^2\E{X}}.
\end{align}
Since
\begin{align}\label{x_e_expectation_upper}
\E \norm{\frac{1}{\tilde \tau_e}\bv{x}_e^{s(e)}}_2^2 = \hat{p}_e \cdot \frac{\tilde \tau_e}{c} \cdot \frac{1}{\tilde \tau_e^2} = \Theta(\log n \epsilon^{-2}),
\end{align}
we can set $\delta = \epsilon$ and conclude that
\begin{align*}
\Pr(\norm{\frac{1}{\tilde \tau_e}\bv{x}_e^{s(e)}}_2^2 > (1+\epsilon)\E \norm{\frac{1}{\tilde \tau_e}\bv{x}_e^{s(e)}}_2^2) = O(n^{-\Theta(1)}).
\end{align*} 
Accordingly, $\norm{\bv{x}_e^{s(e)}}_2^2 \leq c_3 \tilde \tau_e^2 \log n \epsilon^{-2}$ with high probability for some constant $c_3$ and $\epsilon \leq 1$.

Now, if $\bv{x}_e^{s(e)}(e) = 0$, then our sparse recovery routine must return an estimated value for $\bv{x}_e$ that is $\leq \eta \norm{\bv{x}_e^{s(e)}}_2$. We set $\eta = \frac{\epsilon}{c_1 \sqrt{\log n}}$, so with high probability, the returned value is $< \frac{\epsilon}{c_1 \sqrt{\log n}}\sqrt{\norm{\bv{x}_e^{s(e)}}_2^2} = \frac{\sqrt{c_3}\tilde{\tau}_e}{c_1}$. On the other hand, if $\bv{x}_e^{s(e)}(e)$ is non-zero, it equals $\tilde{\tau}_e$, so our sparse recovery sketch must return a value greater than $(1 - \frac{\sqrt{c_3}}{c_1})\tilde{\tau}_e$. Therefore, as long as we set $c_1$ high enough, we can distinguish between both cases by simply checking whether or not the return value is $> \tilde{\tau}_e/2$, as described for Step 2.

Thus, as long as $\norm{\bv{x}_e^{s(e)}}_2^2$ concentrates as described, our procedure recovers $e$ if and only if $\bv{b}_e$ is included in $\bv{B}_{s(e)}$. As explained, this ensures that our process is exactly equivalent to independent sampling. Since concentration holds with probability $O(n^{-\Theta(1)})$, we can adjust constants and union bound over all ${n \choose 2}$ possible edges to claim that our algorithm returns the desired $\bv{\tilde K_\epsilon}$ with high probability.

\end{proof}



\section{Sparsification of Weighted Graphs}\label{weighted}

We can use a standard technique to extend our result to streams of weighted graphs in which an edge's weight is specified at deletion, matching what is known for cut sparsifiers in the dynamic streaming model \cite{gssss,goel2012single}. Assume that all edge weights and the desired approximation factor $\epsilon$ are polynomial in $n$, then we can consider the binary representation of each edge's weight out to $O(\log n)$ bits. For each bit of precision, we maintain a separate unweighted graph $G_0, G_1, ...G_{O(\log n)}$. We add each edge to the graphs corresponding to bits with value one in its binary representation. When an edge is deleted, its weight is specified, so we can delete it from these same graphs. Since G = $ \sum_i 2^i \cdot G_i$, given a $(1 \pm \epsilon)$ sparsifier $ \bv{\tilde K}_i$ for each $\bv{K}_i$ we have:
\begin{align*}
(1-\epsilon) \sum_i 2^i \cdot \bv{K}_i \preceq  \sum_i 2^i \cdot \bv{\tilde K}_i \preceq (1+\epsilon) \sum_i 2^i \cdot \bv{K}_i\\
(1-\epsilon) \bv{K} \preceq  \sum_i 2^i \cdot \bv{\tilde K}_i \preceq (1+\epsilon) \bv{K}.
\end{align*}
So $ \sum_i 2^i \cdot \bv{\tilde K}_i$ is a spectral sparsifier for $\bv{K}$, the Laplacian of the weighted graph $G$.

\section{Sparsification of Structured Matrices}\label{structured}

Next, we extend our algorithm to sparsify certain general quadratic forms in addition to graph Laplacians. There were only three places in our analysis where we used that $\mb$ was not an arbitrary matrix. First, we needed that $\mb = \ms \mb_n$, where $\bv{B}_n$ is the vertex edge incidence matrix of the unweighted complete graph on $n$ vertices. In other words, we assumed that we had some dictionary matrix $\mb_n$ whose rows encompass every possible row that could arrive in the data stream. In addition to this dictionary assumption, we needed $\bv{B}$ to be sparse and to have a bounded condition number in order to achieve our small space results. These conditions allow our compression to avoid an $\Omega(n^2\plog(n))$ lower bound for approximately solving regression on general $\R^{m \times n}$ matrices in the streaming model \cite{clarkson2009numerical}.

As such, to handle the general `structured matrix' case, we assume that we have some dictionary $\bs{\mathcal{A}} \in \R^{m \times n}$ containing $m$ rows $\bv{a}_i \in \R^{n}$. We assume that $m = O(\poly(n))$. In the dynamic streaming model we receive insertions and deletions of rows from $\bs{\mathcal{A}}$ resulting in a matrix $\ma = \ms \bs{\mathcal{A}}$ where $\ms \in \R^{m \times m}$ is a diagonal matrix such that $\bv{S}_{ii} \in \{0, 1\}$ for all $i  \in [m]$. Our goal is to recover from an $O(n \plog (m))$ space compression a diagonal matrix $\mw$ with at most $O(n \log (n))$ nonzero entries such that $ \bs{\mathcal{A}}^\top \mw^2  \bs{\mathcal{A}} \approx_\epsilon  \bs{\mathcal{A}}^\top \ms^2  \bs{\mathcal{A}} = \bv{A}^\top \bv{A}$. Formally, we prove the following:

\begin{theorem}[Streaming Structured Matrix Sparsification] \label{main_structured_theorem}
Given a row dictionary $\bs{\mathcal{A}}\in\R^{m\times n}$ containing all possible rows of the matrix $\bv{A}$, there
exists an algorithm that, for any $\epsilon>0$, processes a stream
of row insertions and deletions for $\ma$ in a single
pass and maintains a set of linear sketches of this input in $O\left(\frac{1}{\epsilon^{2}}n\plog(m,\kappa_u)\right)$
space where $\kappa_u$ is an upper bound on the condition number of $\ma^\top\ma$. From these
sketches, it is possible to recover, with high probability, a matrix
$\bv{\tilde{A}}^{\top}\bv{\tilde{A}}$ such that $\bv{\tilde{A}}$
contains only $O(\epsilon^{-2}n\log n)$ reweighted rows of $\ma$
and $\bv{\tilde{A}}^{\top}\bv{\tilde{A}}$ is a $(1\pm\epsilon)$ spectral
sparsifier of $\ma^{\top}\ma$. The algorithm recovers $\tilde A$ in $\poly(m,\epsilon, n, \log \kappa_u)$ time.
\end{theorem}
Note that, when $m,\kappa_u = O(\poly(n))$, the sketch space is $O\left(\frac{1}{\epsilon^{2}}n\plog(n)\right)$.
To prove Theorem \ref{main_structured_theorem}, we need to introduce a more complicated sampling procedure than what was used for the graph case. In Lemma \ref{leverage_score_bound}, for the correctness proof of $\refine$ in Section \ref{streaming_row_sampling}, we relied on the structure of our graph Laplacian and vertex edge incidence matrix to show that $\tilde \tau_{e',e} \le \tilde \tau_e$. This allowed us to show that the norm of a sampled $\bv{x}_e^{s(e)}$ concentrates around its mean. Thus, we could recover edge $e$ with high probability if it was in fact included in the sampling $\bv{B}_{s(e)}$. Unfortunately, when processing general matrices, $\tilde \tau_e$ is not necessarily the largest element $\bv{x}_e^{s(e)}$ and the concentration argument fails.

We overcome this problem by modifying our algorithm to compute more sketches. Rather than computing a single $\mpi \bv{A}_s$, for every sampling rate $1/2^s$, we compute $O(\log n)$ sketches of different samplings of $\bv{A}$ at rate $1/2^s$. Each sampling is fully independent from the \emph{all} others, including those at the same and different rates. This differs from the graph case, where $\bv{B}_{1/2^{s+1}}$ was always a subsampling of $\bv{B}_{1/2^{s}}$ (for ease of exposition). Our modified set up lets us show that, with high probability, the norm of $\bv{x_i}^{s(i)}$ is close to its expectation for at least a $(1-\epsilon)$ fraction of the independent samplings for rate $s(i)$. We  can recover row $i$ if it is present in one of the `good' samplings. 

Ultimately, we argue, in a similar manner to \cite{KP12}, that we can sample rows according to some distribution that is close to the distribution obtained by independently sampling rows according to leverage score. Using this primitive, we can proceed as in the previous sections to prove Theorem~\ref{main_structured_theorem}. In Section~\ref{sec:sub:gen:row}, we provide the row sampling subroutine and in Section~\ref{sec:sub:gen:sparse}, we show how to use this sampling routine to prove Theorem~\ref{main_structured_theorem}. 

\subsection{Generalized Row Sampling}
\label{sec:sub:gen:row}
Our leverage score sampling algorithm for the streaming model is as follows:

\newcommand{\strlev}{\tilde{\tau}}
\newcommand{\vxst}{\bv{x}_s^{(t)}}
\newcommand{\vxsit}{\bv{x}_{s_i}^{(t)}}
\newcommand{\vxsiti}{\bv{x}_{s_i}^{(t_i)}}
\newcommand{\xsit}{\bv{x}_{s_i}^{(t)}}
\newcommand{\xsiti}{\bv{x}_{s_i}^{(t_i)}}

\begin{framed}{\noindent\bfseries Streaming Row Sampling Algorithm} 

\paragraph{$\maintainma(\bv{A}, \epsilon, \kappa_u, \gamma,c)$:}
\begin{enumerate}

\item Let $S = O(\log \kappa_u)$, $T = O(\log m)$, and for all $s \in [S]$ and $t \in [T]$ let $\bv{F}_s^{(t)} \in \R^{m \times m}$ be a diagonal matrix with $[\bv{F}_s^{(t)}]_{ii} = 1$ independently with probability $\frac{1}{2^s}$ and is $0$ otherwise.\footnotemark
\item For all $s \in [S]$ and $t \in [T]$ maintain sketch $\bv{\Pi}_s^{(t)} \bv{F}_s^{(t)} \bv{A}$ where each $\bv{\Pi}_s^{(t)}$ is drawn independently from the distribution in Lemma~\ref{sparse_recovery_primitive} with $\eta^2 = \frac{1}{C}$ and $C = c_1 \epsilon^{-3} \log m \log n$.
\item Add rows of $\gamma \bv{I}$, independently sampled at rate $\frac{1}{2^s}$ , to each sketch.

\end{enumerate}

\paragraph{$\sample(\bv{\Pi} \ma, \tilde{\mk}, \epsilon,c)$:}
\begin{enumerate}
\item For all $s \in [S]$ and $t \in [T]$ let $\vxst = \bv{F}_s^{(t)} \bv{A} \bv{\tilde K}^+$ and compute $\bv{\Pi}_s^{(t)} \vxst$.

\item For every $i \in [m]$:
\begin{enumerate}[(a)]
\item Compute $\strlev_i = \bv{a}_i^\top \bv{\tilde K}^+ \bv{a}_i$ and $p_i =  c_2\tilde \tau_i \log n \epsilon^{-2}$, where $c_2$ is the oversampling constant from Lemma \ref{sparsifier_sampling}. Choose $s_i$ such that $ \min \{1, p_i \} \le \frac{1}{2^s_i} \le \min \{1, 2p_i \}$.

\item Pick $t_i \in [T]$ uniformly at random and use Lemma~\ref{sparse_recovery_primitive} to check if $\bv{x}_{s_i}^{(t_i)}(i)^2 \geq C^{-1}\norm{\bv{x}_{s_i}^{(t_i)}}_2^2$.

\item If $i$ is recovered, add row $i$ to the set of sampled edges with weight $2^{s_i}$.
\end{enumerate}

\end{enumerate}

\end{framed}\footnotetext{Throughout this section, for $X \in \mathbb{Z}^+$ we let $[X] = \{0, 1, 2, \ldots, X\}$}

We claim that, with high probability, the set of edges returned by the above algorithm is a random variable that is stochastically dominated by the two random variables obtained by sampling edges independently at rates $p_i$ and $(1-\epsilon)p_i$, respectively.

The following property of PSD matrices is used in our proof of correctness:
\begin{lemma}
\label{lem:gen:psd_fact}
For any symmetric PSD matrix $\mk \in \R^{n \times n}$ and indices $i ,j \in [n]$ we have
\[
\left| \mk_{ij} \right| \leq \frac{1}{2} \left(\mk_{ii} + \mk_{jj}\right).
\]
\end{lemma}

\begin{proof}
Let $\bv{1}_i$ be the vector with a $1$ at position $i$ and $0$s else where. For all $i, j \in [n]$ by the fact that $\mk$ is PSD we have that
\[
\left(\bv{1}_i - \bv{1}_j\right) \mk \left(\bv{1}_i - \bv{1}_j\right)
\geq 0
\enspace \text{ and } \enspace
\left(\bv{1}_i + \bv{1}_j\right) \mk \left(\bv{1}_i + \bv{1}_j\right)
\geq 0.
\]
Expanding, we have that:
\[
-\mk_{ii} - \mk_{jj}
\leq 2 \mk_{ij}
\leq \mk_{ii} + \mk_{jj},
\]
yielding the result.
\end{proof}

We can now proceed to prove that our sampling procedure approximates sampling the rows of $\bv{A}$ by their leverage scores.
\begin{lemma}\label{matrix_dominance} Consider an execution of $\sample(\bv{\Pi} \ma, \tilde{\mk}, c, \epsilon)$ where
\begin{itemize}
  \item $c\ma^\top \ma \preceq \tilde{\mk} \preceq \ma^\top \ma$ for $c \in (0, 1]$, and
  \item $\epsilon \in (0, 1]$.
\end{itemize}
Let $\dist$ be a random variable for the indices returned by $\sample(\bv{\Pi} \ma, \tilde{\mk}, c, \epsilon)$. Let $\mathcal{I} \subseteq [m]$ denote the indices of the nonzero rows of $\ma$ and let $\dist_r$ and $\dist_q$ be random variables for the subset of $[m]$ obtained by including each $i \in \mathcal{I}$ independently with probability
\[
r_i = (1 - \epsilon) \frac{1}{2^{s_i}}
~ \text{and} ~
q_i = \frac{1}{2^{s_i}} .
\]
With high probability, i.e. except for a $(1-\frac{1}{m^{O(1)}})$ fraction of the probability space, $\dist$ is stochastically dominated by $\dist_q$ and $\dist$ stochastically dominates $\dist_r$ with respect to set inclusion.
\end{lemma}

\begin{proof}
By definition, $\dist_r$ and $\dist_q$ are always subsets of $\mathcal{I}$ and $\dist$ is a subset of $\mathcal{I}$ with high probability (it is a subset as long as the algorithm of Lemma~\ref{sparse_recovery_primitive} succeeds). Thus it remains to show that, with high probability for each  $\mathcal{J} \subseteq \mathcal{I}$,
\[
\prod_{i \in \mathcal{J}} r_i
= \Pr[\mathcal{J} \subseteq \dist_r]
\leq \Pr[\mathcal{J} \subseteq \dist]
\leq \Pr[\mathcal{J} \subseteq \dist_q]
= \prod_{i \in \mathcal{J}} q_i
.
\]
    Furthermore, by definition, with high probability, $\sample$ outputs $i\in \mathcal{I}$ if and only if $\xsiti(i)^2 \geq C^{-1} \norm{\vxsiti}_2^2$ and consequently
\begin{equation}
\label{eq:gen:sample:1}
\Pr[\mathcal{J} \subseteq \dist]
= \Pr\left[\forall i \in \mathcal{J} ~ : ~ 
\xsiti(i)^2 \geq C^{-1} \norm{\vxsiti}_2^2
\right].
\end{equation}
As shown in Equation \ref{entry_is_lev_score}, when proving our graph sampling Lemma, for all $i \in \mathcal{J}$,
\[
\xsiti(i)
=  [\bv{F}_{s_i}^{(t_i)}]_{ii} \cdot \tilde{\tau}_i .
\]
Consequently, by the definition of $[\bv{F}_{s_i}^{(t_i)}]_{ii}$ we can rewrite \eqref{eq:gen:sample:1} as:
\begin{equation}
\label{eq:gen:sample:2}
\Pr[\mathcal{J} \subseteq \dist]
= \Pr\left[\forall i \in \mathcal{J} ~ : ~
\norm{\vxsiti}_2^2 \leq C \cdot \tilde{\tau}_i^2
~ \text{and} ~
[\bv{F}_{s_i}^{(t_i)}]_{ii} = 1
\right].
\end{equation}
From \eqref{eq:gen:sample:2} and the independence of $[\bv{F}_{s_i}^{(t_i)}]_{ii}$ we obtain the following trivial upper bound on $\Pr[\mathcal{J} \subseteq \dist]$,
\[
\Pr[\mathcal{J} \subseteq \dist]
\leq
\Pr\left[\forall i \in \mathcal{J} ~ : ~ [\bv{F}_{s_i}^{(t_i)}]_{ii} = 1\right]
= \prod_{i \in \mathcal{J}} \frac{1}{2^{s_i}}
= \prod_{i \in \mathcal{J}} q_i
\]
and consequently $\dist$ is stochastically dominated by $\dist_q$ as desired.

As shown in Equation \ref{real_eqn}, when proving the graph sampling case, for all $i \in \mathcal{I}$ and $t \in [T]$
\begin{equation}
\label{eq:gen:sample:3}
p_i \tilde \tau_i \leq \E\left[\norm{\vxsit}_2^2\right]
\leq 
\frac{2}{c} p_i \tilde \tau_i .
\end{equation}
Recalling that $p_i =  c_2\tilde \tau_i \log n \epsilon^{-2}$, combining \eqref{eq:gen:sample:2} and \eqref{eq:gen:sample:3} yields:
\begin{equation}
\label{eq:gen:sample:4}
\Pr[\mathcal{J} \subseteq \dist]
\geq \Pr\left[\forall i \in \mathcal{J} ~ : ~
\norm{\vxsiti}_2^2 \leq c_3 \log m \epsilon^{-1} \cdot \E[\norm{\vxsiti}_2^2]
\text{ and }
[\bv{F}_{s_i}^{(t_i)}]_{ii} = 1
\right],
\end{equation}
where $c_3 = c_1c /2c_2$.

To bound the probability that $\norm{\vxsiti}_2^2 \leq c_3 \log m \epsilon^{-1}  \cdot \E[\norm{\vxsiti}_2^2]$ we break the contribution to $\norm{\vxsit}_2^2$ for each $t$ into two parts. For all $i$ we let $\mathcal{K}_i = \{j \in \mathcal{I} | s_j = s_i\}$, i.e. the set of all rows $j$ which we attempt to recover at the same sampling rate as $i$. For any $t \in [T]$, we let $A_i^{(t)} = \sum_{j \in \mathcal{K}} x_i^{(t)}(j)^2$ and  $B_i^{(t)} = \sum_{j \in \mathcal{I} - \mathcal{K}} x_i^{(t)}(j)^2$. Using this notation and $\eqref{eq:gen:sample:4}$ we obtain the following lower bound
\[
\Pr[\mathcal{J} \subseteq \dist]
\geq \Pr\left[\forall i \in \mathcal{J} ~ : ~
A_i^{(t_i)} \leq \frac{c_3 \log m \epsilon^{-1}}{2} \cdot \E[\norm{\vxsiti}_2^2] ~ \text{,} ~
B_i^{(t_i)} \leq \frac{c_3 \log m \epsilon^{-1}}{2} \cdot \E[\norm{\vxsiti}_2^2] ~ \text{,} ~
\text{and} ~
[\bv{F}_{s_i}^{(t_i)}]_{ii} = 1
\right].
\]
For all $j \in \mathcal{K}_i$, the rows that we attempt to recover at the same rate as row $i$, we know that $\tilde{\tau}_j \leq 2 \tilde{\tau}_i$. By Lemma~\ref{lem:gen:psd_fact} we know that for all $i \in \mathcal{I}$ with $s_i \geq 1$ and $j \in \mathcal{K}_i$
\begin{equation}
\label{eq:gen:sample:6}
x_{s_i}^{(t)}(j)^2
= [\bv{F}_{s_i}^{(t)}]_{jj} \cdot \left| \bv{a}_i^\top \tilde{\mk}^+ \bv{a}_j \right|^2
\leq 1 \cdot \left(\frac{\tilde{\tau}_i + \tilde{\tau}_j}{2}\right)^2
\leq \left(\frac{\tilde{\tau}_i + 2 \tilde{\tau}_i}{2}\right)^2
\leq 3c_2^{-1}\epsilon^2 \log^{-1} n  \E[\norm{\bv{x}_{s_i}^{(t)}}_2^2]
\enspace.
\end{equation}

Now recall that $C = \Omega(\epsilon^{-2} \log n)$. If $\strlev_i > 1/2$ and therefore $s_i = 0$ then $\xsiti(i)^2 = \strlev_i^2$ and setting constants high enough and considering $\eqref{eq:gen:sample:3}$, we see that row $i$ is output with high probability. On the other hand if $s_i \geq 1$, then by \eqref{eq:gen:sample:6} and Chernoff bound choosing a sufficiently large constant we can ensure that with high probability $A_i^{(t)} \leq \frac{c_3 \log m \epsilon^{-1}}{2} \E[\norm{\bv{x}_{s_i}^{(t)}}_2^2]$ for all $i$ and $t$. 

Furthermore, by \eqref{eq:gen:sample:3} and Markov bound we know that $\Pr[B_i^{(t_i)} > \frac{c_3 \log m \epsilon^{-1}}{2}\E[\norm{\bv{x}_{s_i}^{(t)}}_2^2] \leq \frac{\epsilon}{O(\log m)}$. Therefore, by Chernoff bound, with high probability for each $i \in \mathcal{J}$ with $s_i \geq 1$ for at least a $1 - \epsilon$ fraction of the values of $t \in T$ we have $B_i^{(t_i)} \leq \frac{c_3 \log m \epsilon^{-1}}{2}\E[\norm{\bv{x}_{s_i}^{(t)}}_2^2]$. However, note that by construction all the $B_i^{(t)}$ are mutually independent of the $A_i^{(t)}$ and the values of $[\bv{F}_{s_j}^{(t)}]_{jj}$ for $j \in K_i$. So, $\sample$ is simply picking each row $i$ with probability $\frac{1}{2^{s_i}}$ (failing with only a $\frac{1}{m^{O(1)}}$ probability) or not being able to recover each edge independently with some probability at most $\epsilon$ -- the probability that $B_i^{(t_i)}$ is too large. Consequently,  except for a negligible fraction of the probability space we have that
\[
\Pr\left[\mathcal{J} \subseteq \dist\right]
\geq \prod_{i \in \mathcal{J}} (1 - \epsilon) \cdot [\mf_{s_i}^{t}]_{ii}
= \prod_{i \in \mathcal{J}} \frac{1 -\epsilon}{2^{s_i}}
= \prod_{i \in \mathcal{J}} r_i
\]
and we have the desired result.
\end{proof}

\subsection{Generalized Recursive Sparsification}
\label{sec:sub:gen:sparse}

Next we show how to construct a spectral sparsifier in the streaming model for a general
structured matrix using the row sampling subroutine, $\sample$.
In the graph case, Theorem \ref{main_sparsification_theorem} shows
that, if we can find a sparsifier to a graph $G$ using a coarse sparsifier,
then we can use the chain of spectrally similar graphs provided
in Theorem \ref{miller_peng_chain} to find a final $(1 \pm \epsilon)$ sparsifier for our input graph.

The proof of Theorem \ref{main_sparsification_theorem} includes our third reliance on the fact that we are sparsifying graphs -- we claim that the condition number of an unweighted
graph is polynomial in $n$. This fact does not hold in the general matrix case since the condition number can be exponentially large even for bounded integer matrices. Therefore,
our result for general matrix depends on the condition number of $\ma$. 

\begin{theorem}\label{refinement_unweighted_matrix} Given a row
dictionary $\bs{\mathcal{A}}\in\R^{m\times n}$. Let $\ma=\ms\bs{\mathcal{A}}$
be the matrix specified by an insertion-deletion stream where $\ms\in\R^{m\times m}$
is a diagonal matrix such that $\bv{S}_{ii} \in\{0,1\}$ for all $i\in[m]$. Let $\kappa_u$ be a given upper bound on the possible condition number of any $\bv{A}$.
Let $\gamma$ be a fixed parameter and consider $\bv{K}=\ma^{\top}\ma+\gamma\bv{I}$.
For any $\epsilon>0$, there exists a sketching procedure $\maintainma(\bv{A},\epsilon, \gamma, \kappa_u, c)$ that outputs an $O(n \plog (m, \kappa_u))$ sized sketch $\bv{\Pi}\bv{A}$. There exists a corresponding recovery algorithm \texttt{RefineMatrixSparsifier} such that if $c\bv{K} \preceq\bv{\tilde{K}} \preceq \bv{K}$ for some $0 < c < 1$ then:\\

$\texttt{RefineMatrixSparsifier}(\bv{\Pi}\ma,\bv{\tilde{K}},\epsilon,c)$
returns, with high probability, $\bv{\tilde{K}}_{\epsilon}=\bv{\tilde{A}}_{\epsilon}^{\top}\bv{\tilde{A}}_{\epsilon} + \gamma \bv{I}$,
where $(1-\epsilon)\bv{K}\preceq_{r}\bv{\tilde{K}}_{\epsilon}\preceq_{r}(1+\epsilon)\bv{K}$,
and $\bv{\tilde{A}}_{\epsilon}$ contains only $O(\epsilon^{-2}n\log n)$
reweighted rows of $\ma$ with high probability. 
\end{theorem}

\begin{proof} As in the graph case, we can think of the identity $\gamma\bv{I}$ as a set of rows that we sample with probability
$1$. Hence, we have $\bv{\tilde{K}}_{\epsilon}=\bv{\tilde{A}}_{\epsilon}^{\top}\bv{\tilde{A}}_{\epsilon}+\gamma{\bv{I}}$.

Lemma \ref{matrix_dominance} shows that $\sample(\bv{\Pi} \ma,\tilde{\mk},c,\epsilon)$
returns a random set of indices of $\bs{\mathcal{A}}$ such that the generated random variable is
dominated by $\dist_{q}$ and is stochastically dominates $\dist_{r}$. Recall that  $\dist_{r}$ and $\dist_{q}$ are random variables for the subset
of $[m]$ obtained by including each $i\in\mathcal{I}$ independently
with probability 
\[
r_{i}=(1-\epsilon)\frac{1}{2^{s_{i}}}~\text{and}~q_{i}=\frac{1}{2^{s_{i}}}.
\]
Since $\frac{1}{2^{s_{i}}}$ is a constant factor approximation of
leverages score, Lemma \ref{sparsifier_sampling} shows that sampling and reweighing the rows according
to $\dist_{r}$ gives a spectral
sparsifier of $\bv{K}$ with the guarantee required. Similarly, sampling
according to $\dist_{q}$ gives a sparsifier. Since the indices
returned by \textbf{$\sample(\ma,\tilde{\mk},c,\epsilon)$} are sandwiched
between two processes which each give spectral sparsifiers, sampling according
to $\sample$ gives
the required spectral sparsifier \cite{KP12}.
\end{proof}

Using \texttt{RefineMatrixSparsifier}, the arguments in Theorem \ref{main_sparsification_theorem}
yield Theorem \ref{main_structured_theorem}. Our sketch size needs to be based on $\log \kappa_u$ for two reasons -- we must subsample the matrix at $O(\log \kappa_u)$ different rates as our leverage scores will be lower bounded by some $\poly(\kappa_u)$. Further the chain of recursive sparsifiers presented in Theorem \ref{miller_peng_chain} will have length $\log \kappa_u$. Recovery will run in time $\poly(m,n,\epsilon, \log \kappa_u)$. Space usage will depend on the sparsity of the rows in $\bv{A}$ as we will need enough space to solve linear systems in $\bv{\tilde K}$. In the worst case, this will require $O(n^2)$ space, however, if the row of $\bv{A}$ are sparse, and hence $\bv{\tilde K}$ is sparse, recovery will take less space, specifically $O(n \plog(m))$ with constant row sparsity.

\section{Using a Pseudorandom Number Generator}\label{pseudorandomness}
In the proof of our sketching algorithm, Theorem \ref{refinement}, we assume that $\maintain$ has access to $O(\log n)$ uniform random hash functions, $h_1,\ldots,h_{O(\log n)}$ mapping every edge to $\{0,1\}$. These functions are used to subsample our vertex edge incidence matrix, $\bv{B}$, at geometrically decreasing rates. Storing the functions as described would require $O(n^2\log n)$ space - we need $O(\log n)$ random bits for each possible edge. 

To achieve $O(n\plog(n))$ space, we need to compress the hash functions using Nisan's pseudorandom number generator. Our approach follows an argument in \cite{gssss} (Section 3.4) that was originally introduced in \cite{indyk2000stable} (Section 3.3). First, we summarize the pseudorandom number generator from \cite{nisan1992pseudorandom}

\begin{theorem}[Corollary 1 in \cite{nisan1992pseudorandom}]\label{pseudorandomPrimitive}
Any randomized algorithm running in $space(S)$ and using $R$ random bits may be converted to one that uses only $O(S\log R)$ random bits (and runs in space $O(S\log R)$).
\end{theorem}
\cite{nisan1992pseudorandom} gives this conversion explicitly by describing a method for generating $R$ pseudorandom bits from $O(S\log R)$ truly random bits. For any algorithm running in $space(S)$, the pseudorandom bits are ``good enough'' in that the output distribution of the algorithm under pseudorandom bits is very close to the output distribution under truly random bits. In particular, the total variation distance between the distributions is at worst $2^{-O(S)}$ (see Lemma 3 in \cite{nisan1992pseudorandom}). It follows that using pseudorandom bits increases the failure probability of any randomized algorithm by just $2^{-O(S)}$ in the worst case.

As described, our algorithm runs in $O(n^2 \log n)$ space and it is not immediately obvious how to use Theorem \ref{pseudorandomPrimitive} to reduce this requirement.
However, consider the following: suppose our algorithm is used on a sorted edge stream where all insertions and deletions for a single edge come in consecutively. In this case, at any given time, we only need to store one random bit for each hash function, which requires just $O(\log n)$ space. The random bits can be discarded after moving on to the next edge. Thus, the entire algorithm can run in $O(n \plog(n))$ space. Then, we can apply Theorem \ref{pseudorandomPrimitive}, using the pseudorandom generator to get all of our required random bits by expanding just $S\log R = O(n \plog(n)) \cdot O(\log(n^2 \log n)) =  O(n \plog(n))$ truly random bits. 
Since our failure probability increases by at most $1/2^{O(n\plog n)}$, we still only fail with probability inverse polynomial in $n$. 

Now notice that, since our algorithm is sketch based, edge updates simply require an addition to or subtraction from a sketch matrix. These operations commute, so our output will not differ if we reorder of the insertion/deletion stream. Thus, we can run our algorithm on a general edge stream, using the pseudorandom number generator to generate any of the required $O(n^2 \log n)$ bits as they are needed and operating in only $O(n \plog n)$ space.

Each time an edge is streamed in, we need to generate $\log n$ random bits from the pseudorandom generator. This can be done in $\log(R)*S = O(n\plog(n))$ time \cite{indyk2000stable}, which dominates the runtime required to process each streaming update.

Finally, Section \ref{structured} uses a slightly different sampling scheme for general structured matrices. Instead of building a sequence of subsampled matrices, the row dictionary is sampled independently at each level. In total, the required number of random bits is $O(m\log^2n)$, where $m$ is the number of rows in the dictionary $\bv{A}$. We require that $m = \poly(n)$, in which case the arguments above apply unmodified for the general matrix case.




\section{Acknowledgements}
We would like to thank Richard Peng for pointing us to the recursive row sampling algorithm contained in \cite{pengV1}, which became a critical component of our streaming algorithm. 
We would also like to thank Jonathan Kelner for useful discussions and Jelani Nelson for a helpful initial conversation on oblivious graph compression.

This work was partially supported by NSF awards 0843915, 1111109, and 0835652, CCF-1065125, CCF-AF-0937274,  CCF-0939370,  and CCF-1217506, NSF Graduate Research Fellowship grant 1122374, Hong Kong RGC grant 2150701, AFOSR grants FA9550-13-1-0042 and FA9550-12-1-0411, MADALGO center, Simons Foundation, and the Defense Advanced Research Projects Agency (DARPA).

\bibliography{streaming_main}{}
\bibliographystyle{alpha}

\appendix

\section{Sparse Recovery}\label{sparse_recovery_appendix}
In this section we give a proof of the $\ell_2$ heavy hitters algorithm given in Lemma \ref{sparse_recovery_primitive}. It is known that $\ell_2$ heavy hitters is equivalent to the $\ell_2/\ell_2$ sparse recovery problem \cite{gilbert2010sparse}. Some sparse recovery algorithms are in fact based on algorithms for solving heavy hitters problem. However, we were not able to find a suitable reference for an $\ell_2$ heavy hitters algorithm so we show the reduction here - namely, how to find $\ell_2$ heavy hitters using a sparse recovery algorithm. 

We follow the terminology of~\cite{gilbert-stoc2010}. An approximate sparse recovery system consists of parameters $k, N$, an $m \times N$ measurement matrix $\bs{\Phi}$, and a decoding algorithm $D$. For any vector $\bv{x}\in \R^N$ the decoding algorithm $D$ can be used to recover an approximation $\wh{\bv{x}}$ to $\bv{x}$ from the {\em linear sketch} $\bv{\Phi}\bv{x}$. In this paper  we will use a sparse recovery algorithm that achieves the $\ell_2/\ell_2$ sparse recovery guarantee:
$$
||\wh{\bv{x}}-\bv{x}||_2\leq C\cdot ||\bv{x}-\bv{x}_k||_2
$$where $\bv{x}_k$ is the best $k$-term approximation to $\bv{x}$ and $C>1$.  Our main sparse recovery primitive is the following result of~\cite{gilbert-stoc2010}:
\begin{theorem}[Theorem~1 in~\cite{gilbert-stoc2010}]\label{thm:l2l2}
For each $k\geq 1$ and $\epsilon>0$, there is an algorithm and a distribution $\bs{\Phi}$ over matrices in $\R^{O(k\log (N/k) /\epsilon) \times N}$ satisfying that for any $\bv{x}\in \R^N$, given $\bs{\Phi}\bv{x}$, the algorithm returns $\wh{\bv{x}}$ such that $\wh{\bv{x}}$ has $O(k\log^{O(1)} N/\epsilon)$ non-zeros and 
$$
||\wh{\bv{x}}-\bv{x}||_2^2\leq (1+\epsilon) ||\bv{x}-\bv{x}_k||_2^2
$$
with probability at least $3/4$.
The decoding algorithm runs in time $O(k\log^{O(1)} N/\epsilon)$.
\end{theorem}

Using this primitive, we can prove Lemma \ref{sparse_recovery_appendix}.

\begin{replemma}{sparse_recovery_primitive}[$\ell_2$ Heavy Hitters]
For any $\eta >0$, there is a decoding algorithm $D$ and a distribution on matrices $\bv{\Phi}$ in $\R^{O(\eta^{-2} \plog( N)) \times N}$ such that, for any $\bv{x}\in \R^N$, given $\bs{\Phi} \bv{x}$, the algorithm $D$ returns a vector $\bv{w}$ such that
$\bv{w}$ has $O(\eta^{-2}\plog( N))$ non-zeros and satisfies
$$ ||\bv{x}-\bv{w}||_\infty \leq \eta ||\bv{x}||_2.$$
with probability $1-N^{-c}$ over the choice of $\bs{\Phi}$. The sketch $\bv{\Phi} \bv{x}$ can be maintained and decoded in $O(\eta^{-2}\plog( N))$ space. 
\end{replemma}

\begin{proof}
Let $h:[N]\to [16/\eta^2]$ be a random hash function (pairwise independence suffices), and  for $j=1,\ldots, 16/\eta^2$ let $\bv{y}^j_i=\bv{x}_i$ if $h(i)=j$ and $0$ o.w. 
For a vector $\bv{u}\in \R^N$ we write $\bv{u}_{-i}$ to denote $\bv{u}$ with the $i$-th component zeroed out. 

By Markov's inequality we have
$$
\Pr[||\bv{y}^{h(i)}_{-i}||^2>\eta^2 \norm{\bv{x}_{-i}}^2/2]<1/8.
$$
Note that since we are only using Markov's inequality, it is sufficient to have $h$ be pairwise independent. Such a function $h$ can be represented in small space.
Now invoke the result of Theorem~\ref{thm:l2l2} on $\bv{y}^{h(i)}$ with $k=1$, $\epsilon=1$, and let $\bv{w}^{h(i)}$ be the output. We have 
$$
||\bv{y}^{h(i)}-\bv{w}^{h(i)}||_2^2
\leq 2||\bv{y}^{h(i)}-\bv{y}^{h(i)}_k||_2^2
\leq 2||\bv{y}^{h(i)}_{-i}||^2.
$$
Hence, we have
$$ (\bv{y}^{h(i)}_i - \bv{w}^{h(i)}_i)^2 \leq \eta^2 \norm{\bv{x}}^2.$$

This shows that applying sketches from Theorem~\ref{thm:l2l2} to vectors $\bv{y}^{j}$, for $j=1,\ldots, 16/\eta^2$ and outputting the vector $\bv{w}$ with $\bv{w}_i = \bv{w}^{h(i)}_i$ allows us to recover all $i\in [N]$ with $\eta \norm{\bv{x}}_2$ additive error with probability at least $3/4-1/8$.

Performing $O(\log N)$ repetitions and taking the median value of $\bv{w}_i$ yields the result. 
Note that our scheme uses $O(\eta^{-2} \plog(N))$ space and decoding time, and is linear in $\bv{x}$, as desired.
\end{proof}

\section{Recursive Sparsification}\label{miller_peng_appendix}
For completeness, we give a short proof of Theorem \ref{miller_peng_chain}:
\begin{reptheorem}{miller_peng_chain}[Recursive Sparsification -- \cite{pengV1}, Section 4]
Consider any PSD matrix $\bv{K}$ with maximum eigenvalue bounded from above by $\lambda_{u}$ and minimum nonzero eigenvalue bounded from below by $\lambda_{l}$. Let $d = \lceil \log_2 (\lambda_{u}/\lambda_{l})\rceil$.  For $\ell \in \{0,1, 2, ... , d\}$, define:
\begin{align*}
\gamma(\ell) = \lambda_{u}/2^\ell
\end{align*}
So, $\gamma(d) \leq \lambda_{l}$ and $\gamma(0) = \lambda_u$. Then the chain of PSD matrices, $\begin{bmatrix}\bv{K}(0), \bv{K}(1), \ldots, \bv{K}(d)\end{bmatrix}$ with:
\begin{align*}
\bv{K}(\ell) = \bv{K} + \gamma(\ell)\bv{I}_{n\times n}
\end{align*}
satisfies the following relations:
\begin{enumerate}
  \item $\bv{K} \preceq_r \bv{K}(d) \preceq_r 2\bv{K}$
  \item $\bv{K}(\ell) \preceq \bv{K}(\ell-1) \preceq 2\bv{K}(\ell)$ for all $\ell \in \{1,\ldots, d\}$
  \item $\bv{K}(0) \preceq 2\gamma(0)\bv{I} \preceq 2\bv{K}(0)$
\end{enumerate}
When $\bv{K}$ is the Laplacian of an unweighted graph, $\lambda_{max} < 2n$ and $\lambda_{min} > 8/n^2$ (where here $\lambda_{min}$ is the smallest \emph{nonzero} eigenvalue). Thus the length of our chain, $d = \lceil \log_2 \lambda_{u}/\lambda_{l}\rceil$, is $O(\log n)$.
\end{reptheorem}

\begin{proof}
Relation 1 follows trivially from the fact that $\gamma(d) \leq \lambda_{l}$ is smaller than the smallest nonzero eigenvalue of $\bv{K}$. For any $\bv{x} \perp \ker(\bv{K})$:
\begin{align*}
\bv{x}^\top \bv{K}(d) \bv{x} = \bv{x}^\top \bv{K} \bv{x} + \bv{x}^\top(\gamma(d)\bv{I})\bv{x} \leq \bv{x}^\top \bv{K} \bv{x} + \bv{x}^\top(\lambda_{min}\bv{I})\bv{x} \leq 2 \bv{x}^\top \bv{K} \bv{x}
\end{align*}  
The other direction follows from $\gamma(d)\bv{I} \succeq 0$. Using the same argument, relation 3 follows from the fact that $\gamma(0) \geq \lambda_{max}(\bv{K})$. For relation 2:
\begin{align*}
2\bv{K}(\ell) &= 2\bv{K} + 2\gamma(\ell)\bv{I} = 2\bv{K} + \gamma(\ell-1)\bv{I} \succeq \bv{K}(\ell-1)
\end{align*}
Again, the other direction just follows from  $\gamma(\ell)\bv{I} \succeq 0$. 

Finally, we need to prove the required eigenvalue bounds. For an unweighted graph, $\lambda_{max} < n$ follows from fact that $n$ is the maximum eigenvalue of the Laplacian of the complete graph on $n$ vertices. $\lambda_{min} > 8/n^2$ by Lemma 6.1 of \cite{spielmanTengSolver}.
Note that this argument extends to weighted graphs when the ratio between the heaviest and lightest edge is bounded by a polynomial in $n$.
\end{proof}

\end{document}